\documentclass[letterpaper, 10pt, conference]{ieeeconf}  
\IEEEoverridecommandlockouts      
\usepackage{epsfig}
\usepackage{amssymb}
\usepackage{amsmath}
\usepackage{amsfonts}
\usepackage{graphicx}
\usepackage[figure,lined,boxed,commentsnumbered]{algorithm2e}
\usepackage{color}
\usepackage{verbatim}
\allowdisplaybreaks              
{
	\newtheorem{assumption}{\textbf{Assumption}}
}
{
	
}
{

}
{

	\newtheorem{lemma}{\textbf{Lemma}}
}
{
	
}

\newtheorem{example}{\textbf{Example}}
\begin{document}
	\pagestyle{plain}
	\title{\textbf{Decomposition of Nonlinear Dynamical Systems Using Koopman Gramians}}
	\author{
		Zhiyuan Liu,\thanks{Z. Liu and L. Chen are with the Department of Computer Science, University of Colorado, Boulder, CO 80309, USA (email: \{zhiyuan.liu, lijun.chen\}@colorado.edu),}
		Soumya Kundu, Lijun Chen, and
		Enoch Yeung \thanks{S. Kundu and E. Yeung are with the Pacific Northwest National Laboratory, Richland, WA 99354, (email: \{soumya.kundu, enoch.yeung\}@pnnl.gov).
	}
    }
	\maketitle
	\begin{abstract}
        In this paper we propose a new Koopman operator approach to the decomposition of nonlinear dynamical systems using Koopman Gramians. We introduce the notion of an input-Koopman operator, and show how input-Koopman operators can be used to cast a nonlinear system into the classical state-space form, and identify conditions under which input and state observable functions are well separated.  We then extend an existing method of dynamic mode decomposition for learning Koopman operators from data known as deep dynamic mode decomposition to systems with controls or disturbances.  We illustrate the accuracy of the method in learning an input-state separable Koopman operator for an example system, even when the underlying system exhibits mixed state-input terms.  We next introduce a nonlinear decomposition algorithm, based on Koopman Gramians, that maximizes internal subsystem observability and disturbance rejection from unwanted noise from other subsystems. We derive a relaxation based on Koopman Gramians and multi-way partitioning for the resulting NP-hard decomposition problem.  We lastly illustrate the proposed algorithm with  the swing dynamics  for an IEEE 39-bus system.            
	\end{abstract}
 	\section{Introduction}
The design and control of complex systems is usually broken down into simpler modules, e.g. different geographically dispersed subsystems such as different balancing authority areas in power networks and various vertically integrated functionalities such as routing and congestion control in communication networks. Despite modular design being a common practice, the decomposition of the system into different modules is often based on engineering intuition rather than principled methodologies. 



For example, in some applications, the physical layout or construction of a system dictates the subsystem structure.  For example, in multi-cellular biological systems, the subsystems are a natural consequence of physical separation by the membrane barrier \cite{regot2011distributed, tamsir2011robust,miller2001quorum}.  In critical infrastructure systems, e.g. the power grid or water distribution systems, the subsystem structure is traditionally defined based on distance and connectivity of buses, nodes or junctions \cite{chiang2016large,li2008interdependency}.  However, this choice of decomposition is critical since a poor system or model decomposition can introduce fundamental limits into distributed controller performance.

Moreover, there are many scenarios where a suitable system decomposition may not be known {\it a priori}, e.g. design of large-scale or ad-hoc communication networks \cite{chiang2016fog, chiang2007layering, chiang2016large} or cyber-physical systems made of agile teams of agents.  In such situations, the system decomposition must be identified, using appropriate criteria that enable performance of a distributed controller.  This is especially true of extremely large scale networks involving thousands of variables where synthesis of a global controller may not be computationally feasible.   

There are many existing methods for model decomposition \cite{sanchez2014hierarchical,raak2014partitioning,chiang2007layering};   each method decomposes systems based on different properties.  Sanchez-Garcia et al. \cite{sanchez2014hierarchical} decomposes a system into subsystems by examining the spectrum of a matrix encoding connectivity in power networks.  When that matrix is the admittance matrix, the decomposition reflects static connectivity structure; when the matrix is power flow, the decomposition reveals islands with minimal power flow disruption.  Chiang, Low, and Doyle showed that system decomposition for layered network design problems can be thought of as vertical and horizontal decomposition of optimization problems with utility functions corresponding to different Quality of Service metrics, e.g. fairness, minimal congestion, or efficiency \cite{chiang2007layering}.   Raak et al. \cite{raak2014partitioning} use the Koopman operator to define a lifted linear representation of an open-loop nonlinear system.  They introduce a decomposition method based on the point-spectrum of the Koopman operator.   

In this paper we adopt a similar Koopman operator approach for system decomposition as Raak et al. \cite{raak2014partitioning}, but derive the decomposition from input-output properties directly computed from Koopman Gramians \cite{YeungGramians}.  Koopman operators are especially powerful, since they allow us to transform nonlinear analysis and control problems into linear problems.  Furthermore, system identification methods for Koopman operators have made it possible to perform analysis and control design in a purely data-driven fashion \cite{mezic2004comparison,mezic2005spectral,rowley2009spectral,williams2015data}.  Such an approach is especially valuable in data-driven control and analysis problems, e.g., Internet of Things applications \cite{chiang2016fog}, intelligent load control problems in infrastructure systems \cite{trudnowski2006power}, or design of multicellular biocircuits \cite{tamsir2011robust}.  

Specifically, we introduce the notion of an input-Koopman operator, show how input-Koopman operators can be used to cast a nonlinear system into the classical state-space form, and identify conditions under which input and state observable functions are well separated.  We then extend an existing method of dynamic mode decomposition for learning Koopman operators from data known as deep dynamic mode decomposition to systems with controls or disturbances.  We illustrate the accuracy of the method in learning an input-state separable Koopman operator for an example system, even when the underlying system exhibits mixed state-input terms.  We next introduce a nonlinear decomposition algorithm, based on Koopman Gramians, that maximizes internal subsystem observability and disturbance rejection from unwanted noise from other subsystems.  As the resulting decomposition problem is NP-hard, we derive a relaxation based on Koopman Gramians and multi-way partitioning.  We illustrate the proposed algorithm with the swing dynamics for am IEEE 39-bus system        

The rest of the paper is organized as follows.  Section II introduces input-Koopman operators and reviews the theory of input-Koopman operators.  Section III discusses several results and techniques for learning Koopman operators and the associated Koopman Gramians.  Section IV presents the proposed Koopman decomposition algorithm, and Section V illustrates its application on an IEEE 39-bus system.


	\section{Input-Koopman Operators }
    
In \cite{proctor2016generalizing,proctor2016dmdcontrol}, it was shown that the Koopman operator could be generalized to model the effect of inputs or controls.  The authors showed that an input-Koopman operator $K$ could be defined that satisfies 
\begin{equation}\label{eq:inputKoopman}
\begin{aligned}
 \psi (x_{t+1},w_{t+1})& =  K \circ \psi(x_{t}, w_{t}) \\  &= \psi(f(x_t,w_t)).
\end{aligned}
\end{equation}
where $\psi(x_t,w_t)\in \mathbb{R}^n_L, n_L \leq \infty$ is a potentially infinite-dimensional  dictionary of observables defined on the state $x_t$ and the input $w_t.$  Thus, $K$ is a linear operator with a well-defined spectrum, either discrete and countable, or continuous and uncountably infinite.  In this paper, we consider the scenario where $K$ is a finite or countably infinite dimensional operator.
\begin{assumption}\label{assump:countableKoopman}
We suppose that the system (\ref{eq:nonlinear_affine_system}) has a finite or countably infinite dimensional Koopman operator.
\end{assumption}
\indent If the control inputs have their own dynamics, e.g. they are defined by the state-space dynamics of a controller, then $\psi(x_{t+1},w_{t+1})$ defines a linear state-space model for $x_{t}$ and $w_{t}$.  In the case where $w_t$ has no state-dynamics, e.g. it is modeled as an exogeneous input or random disturbance, then the input-Koopman operator $K$ satisfies 
\begin{equation}
\begin{aligned}
\psi(x_{t+1},0)   &=  K \psi(x_{t}, w_{t})  
\end{aligned}
\end{equation}
For the purposes of this paper, both for defining Koopman gramians and performing input-Koopman analysis, we model inputs as disturbances without state-space dynamics. 
\begin{assumption}\label{assump:noinputdynamics}
Given the Koopman model (\ref{eq:inputKoopman})  for the nonlinear time-invariant system (\ref{eq:nonlinear_affine_system}) \begin{equation}\psi(x_{t+1},w_{t+1}) = \psi(x_{t+1},0) = K \psi(x_t,w_t)\end{equation} 
\end{assumption} In essence, we model inputs as purely exogenous, impulse functions, or random disturbances to the system of interest.  This assumption is equivalent to saying that $w_t$ is not necessary to include into $\psi(\cdot)$ in order to satisfy the axiom of state for the dynamical system (\ref{eq:exoKoopman}).   The consequence of this assumption is that we can express the system in affine Koopman control form.
\begin{lemma}
Consider a nonlinear system of the form 
\begin{equation}\label{eq:nonlinear_system}
x_{t+1} = f(x_t,w_t)
\end{equation}
with exogenous disturbances $w_t$ and corresponding Koopman model satisfying Assumption \ref{assump:noinputdynamics},
\begin{equation}\label{eq:exoKoopman}
\begin{aligned}
 \psi (x_{t+1},0) &=   K \psi(x_{t}, w_{t}).  
\end{aligned}
\end{equation}
The same Koopman equation can be written as
\begin{equation} \label{eq:affineKoopman}
\psi_x(x_{t+1}) = K_x \psi_x(x_t) + K_u \psi_u(u_t) 
\end{equation}
where $u_t = u(x_t,w_t)$ is a vector function consisting of univariate terms of $w_t$ and  multivariate polynomial terms consisting of $x_t$ and $w_t.$
\end{lemma}
\begin{proof}
Consider the nonlinear system (\ref{eq:nonlinear_system}).  We remark the form of equation (\ref{eq:exoKoopman}) is a special instance of the form derived in \cite{proctor2016generalizing}.  To be precise, the existence of a closed-loop system Koopman operator that satisfies the relation
\begin{equation}
\psi(x_{t+1},w_{t+1}) = K \psi(x_t,w_t) 
\end{equation}
follows from the original Koopman papers \cite{koopman1931hamiltonian,koopman1932dynamical}.  The entire state-space dynamics of a closed-loop nonlinear system, including both state and input, can be viewed as the state-space dynamics of an autonomous dynamical system which has a Koopman operator.    Moreover, Assumption \ref{assump:noinputdynamics} guarantees that the system can be written in the form 
\begin{equation}
\psi_x(x_{t+1})  = K \psi(x_t,w_t) 
\end{equation}
where $\psi_x(\cdot)$ is a vector consisting of the elements of $\psi(x_t,w_t)$ that only depend on $x_t.$ Due to Assumption \ref{assump:countableKoopman} we know that $K$ is a linear operator that can be represented by a matrix of countable dimension.  Therefore, the right hand side can be partitioned in terms of dependence of Koopman basis functions on $x_t, w_t$ or both $x_t$ and $w_t$:
\begin{equation}
\psi_x(x_{t+1})  = K_x \psi_x(x_t) + K_{xw}\psi_{xw}(x_t,w_t) + K_w\psi_w(w_t) 
\end{equation}
where $\psi_x(x_t)$ represents the elements of $\psi(x_t,w_t)$ that directly depend on $x_t$, $\psi_{xw}(x_t,w_t)$ represents the elements of $\psi(x_t,w_t)$ that depend on a  mixture of $x_t$ and $w_t$ terms, and $\psi_w(w_t)$ represents the elements of $\psi(x_t,w_t)$ that only depends on $w_t.$ 
Now consider the last two terms on the right hand side; we can write an exact expression according to Taylor's theorem  for each term 
\begin{equation}
\begin{aligned}
\psi_{xw}(x_t,w_t) & = W_{xw} \nu(x_t,w_t) \\ 
\end{aligned}
\end{equation}
where $\nu(x_t,w_t)$ is a vector containing the polynomial basis with elements of the form 
\begin{equation}
x_i^l w_j^k,
\end{equation}
$l, k \in \mathbb{N}$, $x_i$ is an element of the state vector $x$, $i = 1,2 , ...$ and $w_j$ is an element of the disturbance vector $w$, $j = 1, 2, ...$  Similarly, 
$\nu(w_t)$ is a vector containing the polynomial basis with elements of the form 
\begin{equation}
w_i^l w_j^k
\end{equation}
where $i, j = 1,2 , ...$ and $l,k \in \mathbb{N}.$  Define \[u_t =\begin{bmatrix} w_t^T & \nu(x_t,w_t)^T    \end{bmatrix}^T\]
It immediately follows that 
\[K_u = \begin{bmatrix}K_{w} \\ K_{xw} W_{xw} \end{bmatrix}\]
and therefore
\begin{equation}
\psi_x(x_{t+1}) = K_x \psi_x(x_t) + K_u \psi_u(u_t) 
\end{equation}
\end{proof}
Finally, we suppose that the observable functions retain the information of the underlying output $y_t$, via a projection. 
\begin{assumption}
We suppose, as in \cite{williams2015data}, that there exists a $W_h \in\mathbb{R}^{p\times n_L}$ such that \begin{equation*}
	\begin{aligned}
    	y_t = h(P_x \psi_x(x_t)) = W_{h}\psi_x(x_t)
    \end{aligned}
\end{equation*}
is in the span of the observables $\psi_x^1(x_t),\cdots, \psi_{x}^{N}(x_t)$. 
\end{assumption}
This assumption is important for the formulation of a Koopman observability gramian 
\cite{YeungGramians}.  

We now consider a class of nonlinear systems which have state-control separability.  Specifically, we consider systems like (\ref{eq:nonlinear_system}) with the explicit form 
\begin{equation}\label{eq:nonlinear_affine_system}
\begin{aligned}
x_{t+1}  &= f(x_t) + g(w_t)  \\ 
y_t &= h(x_t)
\end{aligned}
\end{equation}
where $f,g,h  \in \mathbb{R}^n$ is assumed to be continuously differentiable functions (with $f,g \in C^p\left[0,\infty\right)$ where $p \geq 2$) with respect to $x$ and $w$, respectively.   Without loss of generality, we suppose that $w_t = 0$ implies that $g(w) = 0$, i.e. any zeroth-order terms in a Taylor expansion of $g(u)$ are subsumed into the definition of $f(x)$ as an affine constant offset.
\begin{lemma}\label{lemma:2}
Consider a state-input separable system of the form (\ref{eq:nonlinear_affine_system}).  Suppose the system satisfies Assumptions \ref{assump:countableKoopman} and \ref{assump:noinputdynamics}.   Let $K_x$ denote a Koopman operator for the corresponding open loop system 
\begin{equation}
\begin{aligned}
x_{t+1} &= f(x_t) \\ 
y_t =& h(x_t) 
\end{aligned}
\end{equation}
such that 
\begin{equation}
\psi_{x}(x_{t+1}) = \psi_x(f(x)) = K_x \circ \psi_{x}(x_t) 
\end{equation}
If there exists functions $\psi_{xw}, \psi_w$ and matrices $K_{xw}, K_w$ such that
\begin{equation}
\begin{aligned}
\psi_{x}(x_{t+1}(x_t,w_t)) - \psi_x(x_{t+1}(x_t,0)) & = K_{xw} \psi_{xw}(x_t,w_t) \\ & \hspace{3mm}+ K_w \psi_w(x_t,w_t) 
\end{aligned}
\end{equation} 
then
\begin{equation}\label{eq:AKD}
\begin{bmatrix} \psi_{x}(x_{t+1}) \\ \psi_{xw}(x_{t+1},w_{t+1}) \\ \psi_{w}(w_{t+1}) \end{bmatrix} =
 \begin{bmatrix}
K_x&K_{xw} & K_w 
\end{bmatrix} \begin{bmatrix} \psi_{x}(x_t) \\ \psi_{xw}(x_t,w_t) \\ \psi_{w}(w_t) \end{bmatrix} 
\end{equation}
and 
\begin{equation}
K = \begin{bmatrix} K_x&K_{xw} & K_w \end{bmatrix} 
\end{equation}
is a Koopman operator for the $w_t$-perturbed system (\ref{eq:nonlinear_affine_system}). 
Finally, if i) $\psi_{xw}(x,w)$ defines a basis set on $x$ even when $w=0$ and  ii) $\psi_x(x)$  is continuous, then \[K_{xw} \equiv 0.\] 
\end{lemma}
\begin{proof}
Suppose that there exists functions $\psi_{xw}, \psi_w$ and matrices $K_{xw}, K_w$ such that
\begin{equation}
\begin{aligned}
\psi_{x}(x_{t+1}(x_t,w_t)) - \psi_x(x_{t+1}(x_t,0)) & = K_{xw} \psi_{xw}(x_t,w_t) \\ & \hspace{3mm}+ K_w \psi_w(x_t,w_t) 
\end{aligned}
\end{equation} 
Since $K_x$ is the open-loop Koopman operator,
\begin{equation}
\begin{aligned}
\psi_x(x_{t+1}(x_t,0)) &= \psi_x(f(x) + g(0)) \\ & = \psi_x(f(x)) \\ & = K_x \psi_x(x_t)
\end{aligned}
\end{equation}
which proves equation (\ref{eq:AKD}).  Next, suppose that $K_{w}$ has full column rank and $\psi_x(x)$ is continuous.  Then taking the limit as $w_t \rightarrow 0$ to obtain 
\begin{equation}
\lim_{w_t \rightarrow 0} \psi_{x}(x_{t+1}(x_t,w_t)) - \psi_x(x_{t+1}(x_t,0)) = 0 
\end{equation}
which implies that 
\begin{equation}
0 =  \lim_{w_t \rightarrow 0}K_{xw} \psi_{xw}(x_t,w_t) + K_w \psi_w(w_t) .
\end{equation}
Now $-K_w\psi_w(0)$ is a constant vector, if it is a non-zero constant vector then we would obtain a contradiction since $\psi_{xw}(x,0)$ is a function in $x$.  Thus, $-K_w\psi_w(0) = 0$, which implies that
\begin{equation}
K_{xw} \psi_{xw}(x_t,0)= -K_w \psi_w(0) =  0
\end{equation}
and since $\psi_{xw}(x_t,0)$ is a basis with respect to $x_t$, this means that the matrix 
\begin{equation}
\begin{bmatrix} \psi_{xw}(x_1,0)&| & \hdots &|&\psi_{xw}(x_S,0)  \end{bmatrix}
\end{equation}
is right invertible as long as $S \geq \dim(\psi_{xw})$, which implies that $K_{xw}\equiv 0$. 
\end{proof}




\section{Deep Dynamic Mode Decomposition for Control-Koopman Operator Learning}
The traditional method for learning a Koopman operator from data is extended dynamic mode decomposition.  In extended dynamic mode decomposition for open-loop systems, a dictionary of empirical observable functions \begin{equation}
\Psi(x) = \{ \psi_1(x), \psi_2(x), ... \psi_{n_D} \}
\end{equation}
is postulated to span a large enough subspace of the true Koopman (and unknown observable functions).

In input-Koopman (or control-Koopman) learning problems \cite{proctor2016generalizing,proctor2016dmdcontrol} for systems satisfying Assumption \ref{assump:noinputdynamics}, a dictionary of empirical observable functions is defined in terms of the system state $x \in \mathbb{R}^n$ and input $w \in \mathbb{R}^m$  
of the form 
\begin{equation}
\Psi(x,u) = \{ \psi_x(x), \psi_{xw}(x,w), \psi_{w}(w) \}
\end{equation}
where $\psi_x(x) \in \mathbb{R}^{n_L}$ and  $(\psi_{xw}(x,w), \psi_w(w) \in \mathbb{R}^{m_L}.$   Again, the challenge is that the true Koopman state and input observable functions are unknown.   

In general, there always exists a trivial Koopman observable function $\psi_x(x) \equiv 0$ and $(\psi_{xw}(x,w), \psi_w(w) ) \equiv 0$  with corresponding trivial control and state Koopman matrices $K_x \equiv 0 $ and $K_{w} \equiv 0.$    Such Koopman representations are not of interest to us, since ultimately we are interested finding concise representations that elucidate underlying system dynamics, in terms of the state $x_t$ and $w_t$.  Thus, we adopt the classical assumption of including the state $x_t$ and disturbance $w_t$ in the dictionary functions. 
\begin{assumption}\label{assump:stateinputinclusive}
Let $\Psi(x,u)$ denote the Koopman observables dictionary.  We suppose that there exists a collection of $\psi_{x,1}(x) , ...,\psi_{x,n}(x), \psi_{w,1}(w), ..., \psi_{w,m} \in \Psi(x,u)$  such that 
\begin{equation}
\left(\psi_{x,1}(x_t) , ..., \psi_{x,n}(x) \right) = x_t \in \mathbb{R}^n
\end{equation}
and 
\begin{equation}
\left(\psi_{w,1}(w_t) , ..., \psi_{w,m}(w_t) \right) =w_t \in \mathbb{R}^m 
\end{equation}

We refer to such a dictionary $\Psi(x,u)$ satisfying these properties as {\it state and input inclusive}. 
\end{assumption}
Given a system (\ref{eq:nonlinear_system}), with time series data $y_1,...,y_t$ the control-Koopman learning problem is to find a set of observable functions $\psi_x(x), \psi_w(w),$ and $\psi_{xw}(x,w)$, and matrices $K_x \in\mathbb{R}^{n_L
\times n_L}, W_h \in\mathbb{R}^{p\times n_L}$ and $\begin{bmatrix} K_{xw} & K_{w} \end{bmatrix} \in\mathbb{R}^{n_L \times m_L}$ to solve the optimization problem
\begin{equation}\label{eq:learningobjective}
\min_{W_h, \psi_x, \psi_w, \psi_{xw}, K_{xw}, K_{w}, K_x } ||Y_{t:1} - F_{t-1:0}|| \end{equation}
where 
\begin{equation}
Y_{t:1} = \begin{bmatrix} y_{t} & y_{t-1} & \hdots & y_{1} \end{bmatrix} 
\end{equation}
and 
\begin{equation}
F_{t-1:0} = \begin{bmatrix} \kappa(t-1) & \hdots \kappa(1) ,\end{bmatrix} 
\end{equation}
with 
\begin{equation}
\begin{aligned}
\kappa(t-1) &= - W_h ( K_x\psi_x(x_{t-1})+ K_{xw}\psi_{xw}(x_{t-1},w_{t-1}) \\ &\hspace{5mm} + K_{w} \psi_w(w_{t-1})).
\end{aligned}
\end{equation}
In general, learning the Koopman observable functions for both the input and state can be computationally expensive.  Typically, a generic but expressive set of dictionary functions such as Hermite polynomials, Legendre polynomials or thin-plate radial basis functions are used \cite{williams2015data}.  However, the number of dictionary functions required is not known a priori and often the number of dictionary terms requires multiple steps of manual refinement, even for the simple two or three state systems.  

Recently, it was shown that deep (and shallow) neural networks can be used to generate Koopman observable dictionaries that automatically update during the training process \cite{yeung2017learning,li2017extended}.  Moreover, the dictionaries learned using neural networks appear to be efficient at encoding Koopman dictionaries for larger systems, e.g. a partially observed large-scale linear system and a glycolytic oscillator with 7 states.  

We extend the methods developed in \cite{yeung2017learning}, to address the input-Koopman operator learning problem.  In particular, we suppose that $\psi_x, \psi_{x,w},$ and $\psi_w$ are the outputs of three separate neural networks that multiply against decision variables $K_x$, $K_{x,w}$ and $K_{w}$ respectively.  The learning objective is thus defined as the Frobenius norm of objective function (\ref{eq:learningobjective}).   
In particular, the deep neural networks allow us to parameterize the Koopman observable functions as follows: 
\begin{equation}\begin{aligned}
\psi_x(x)& = (x,\mathcal{D}_x(x,\theta_x) )\\ 
\psi_w(w)& =  (w,\mathcal{D}_w(w,\theta_w)) \\ 
\psi_{xw}(x,w) &\approx \mathcal{D}_{xw}(x,w,\theta_{xw})
\end{aligned}\end{equation}
with 
\begin{equation}
\begin{aligned}
\theta_x &= (W^{1}_x,...,W^{x_D}_x,b^1_x,...,b^{x_D}_x) ; \\  \theta_{xw} &= (W^1_{xw},...,W^{xw_D}_{xw}, b^1_{xw},...,b^{xw_D}_{xw})\\  \theta_{w} & = (W^1_w,...,W^{w_D}_w,b^1_w,...,b^{w_D}_w)
\end{aligned}
\end{equation}
and for a variable $v = x,w$ or mixed terms from $(x,w)$, 
\begin{equation}
\mathcal{D}_v(v,\theta_v) \equiv h_{v_D} \circ h_{v_{D-1}} \circ \hdots \circ h_{1}(v) 
\end{equation}
and $h_i(v) = \sigma (W^i_{v}v + b^i_v)$ and $\sigma(\cdot)$ is an activation function, e.g. a RELU$(v)$, ELU$(v)$, tanh$(v)$, or cRELU$(v)$.  

\begin{example} {\bf Deep Koopman Learning on an 2 State System with a Single Input}
We first illustrate the use of deep control Koopman learning on a simple two state example system with a single input. Consider the system 
\begin{equation}\label{eq:}
\begin{aligned}
x_{t+1}^1 &= -a_1 x_{t}^2 + a_3 x_t^1\sin(\omega u_t) \\
x_{t+1}^2 &= \sin(\omega x_{t}^1) + a_2 x_t^2 + x_t^1 x_t^2 u_t 
\end{aligned}
\end{equation}
where $a_1 = -0.96, a_2 = 0.88, a_3= -0.95,$ and $\omega = -2.0.$  The input signal is defined as a step function \begin{equation} u_t = \begin{cases}0 & t < 250 \\ 1 & t \geq 250.\end{cases}\end{equation}  A simulation of the system for the initial condition $x_0 = (0.5, -0.1)$ is given in Figure \ref{fig:twostatesexample} and the outcome of a multi-step prediction task is given in Figure \ref{fig:twostatesprediction}.  250 points of training data are provided to the Koopman learning algorithm, while the remaining 250 data points are withheld during training for test and evaluation. In Figure \ref{fig:twostatesprediction}, we provide a single initial condition for which no forward prediction training data was provided and evaluate the predictive capability of the Koopman operator.  We see in Figure \ref{fig:twostatesprediction} that the average error of the  approximate input-Koopman operator learned by the deep neural networks is approximately $1\%$, per time-point, over a 120 step forward prediction task. 
\begin{figure} \label{fig:twostatesexample}
\centering
\includegraphics[width=0.9\columnwidth]{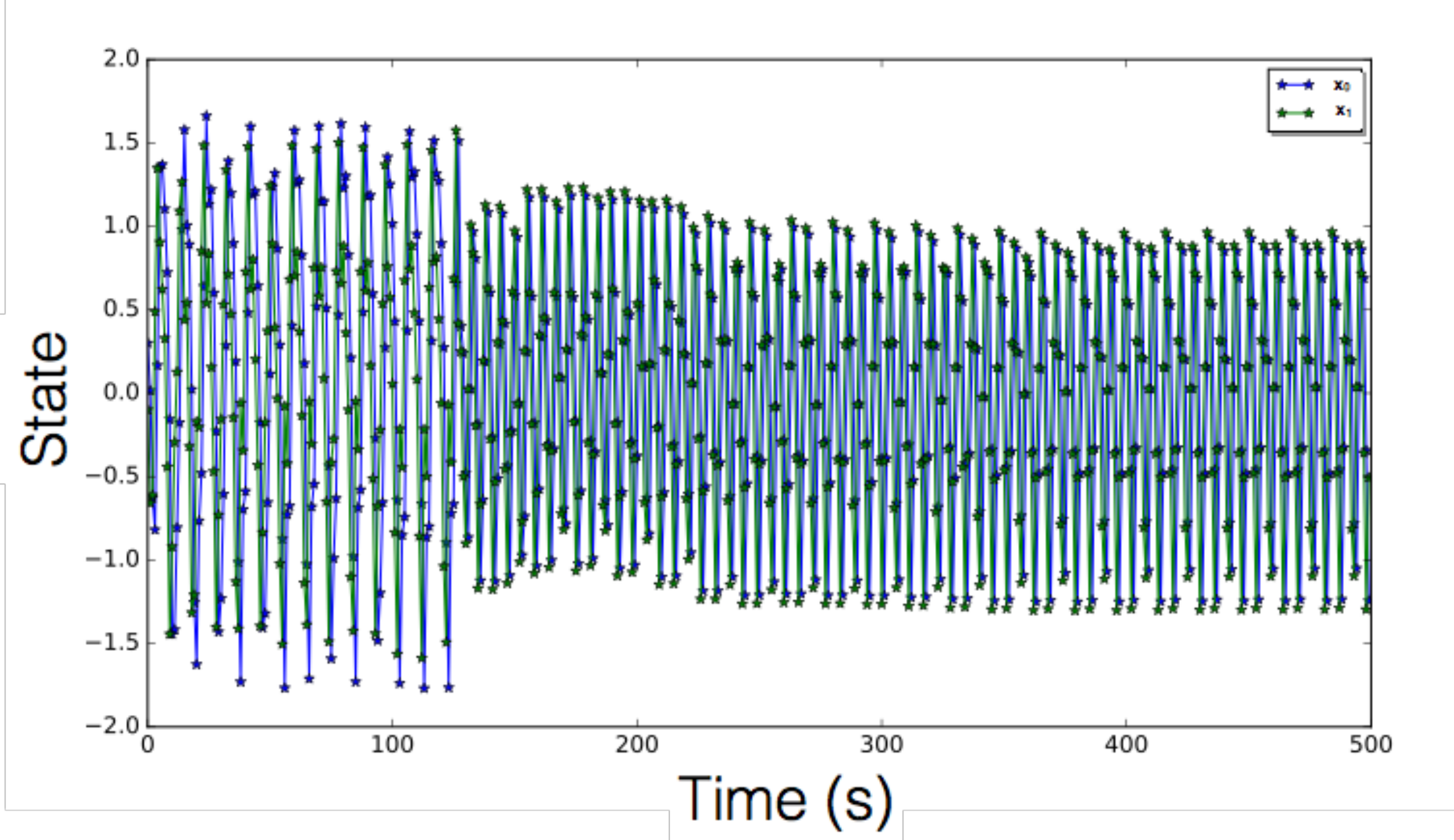}
\caption{The time-lapse response of a two-state control system to initial condition $x_0 = (0.5,-0.1)$ and a step input $u_t$.  The output of the first state and second states, $x_t^1$ and $x^2_t$ are marked with blue and green dots respectively.}
\end{figure}
\begin{figure} \label{fig:twostatesprediction}
\centering
\includegraphics[width=0.9\columnwidth]{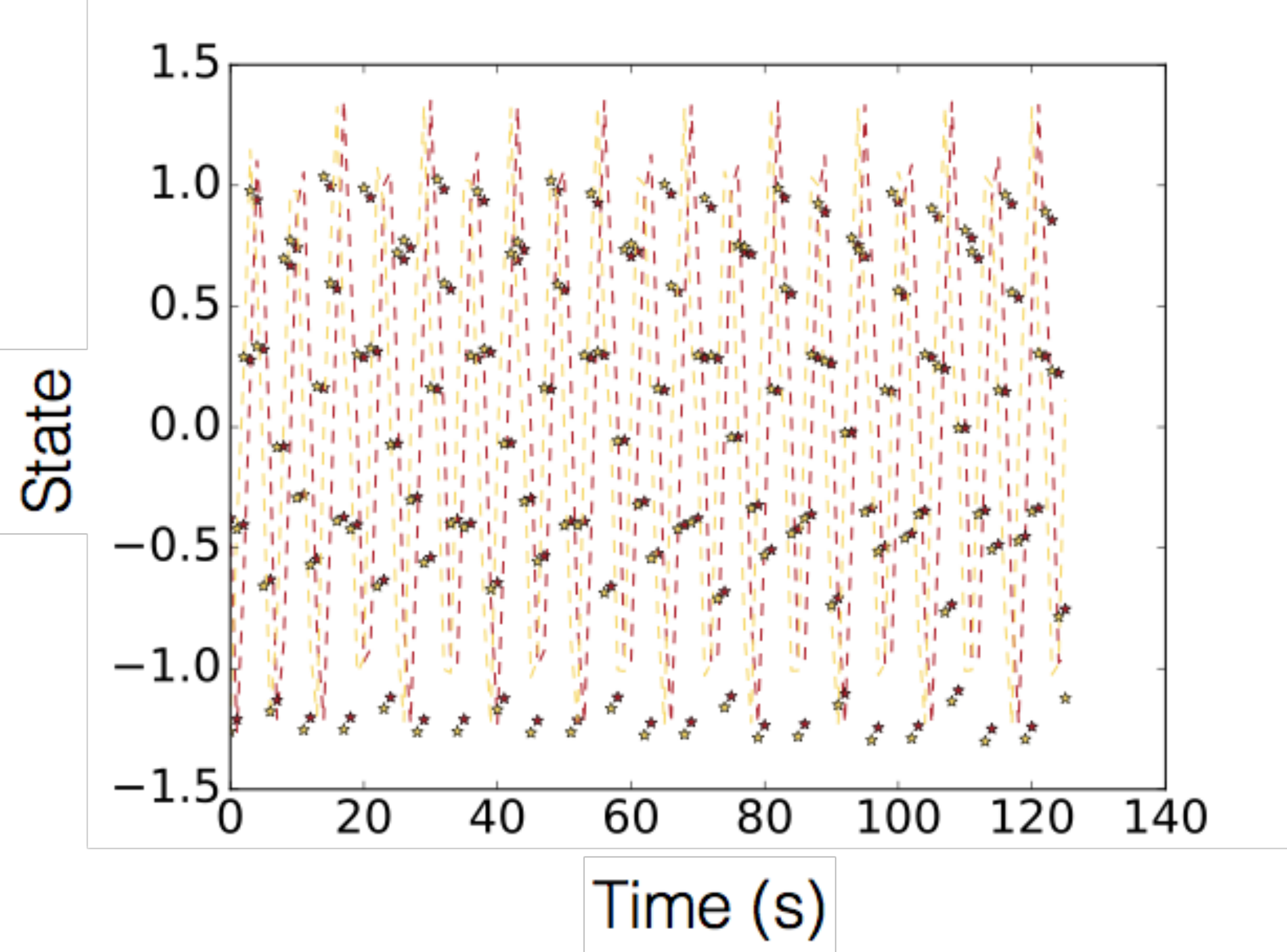}
\caption{Multi-step forward prediction (120 steps, plotted as dashed lines) for the system (plotted as dots) outlined in Example 1 .}
\end{figure}
\end{example}
 
\section{Model Decomposition Algorithm}
We now consider a system decomposition approach for nonlinear dynamical systems, where subsystems are selected to satisfy two criteria: 1) each subsystem's ability to infer its internal states given local output measurements is maximized, 2) the influence of all disturbances $w_t$ from other subsystems on a given subsystem $S_i$ is minimized.

We use the Koopman controllability and observability gramians to compute the decomposition criterion. Following the Koopman operator approach outlined in \cite{YeungGramians}, we can define observability $X_{o}^\psi$ and controllability gramian $X^\psi_{c}$
\begin{equation}
\begin{aligned}
X_o^\psi  &= \sum_{t=0}^\infty W_h^T (K_x^{t})^{T} K_x^{t}W_h,\\
X_c^\psi  &=  \sum_{t=0}^\infty K_x^t K_u K_u^T (K_x^t)^{T}. 
\end{aligned}
\end{equation}
As shown in \cite{YeungGramians}, the Koopman gramians provide a direct approach for quantifying nonlinear controllability and nonlinear observability.  When considering a linear time-invariant system, the Koopman gramians are identical to the original gramians.  Furthermore, when a nonlinear system is locally controllable or observable, there exists a projection of the Koopman gramian that is positive definite.  

In addition, the Koopman gramians characterize the input-output properties of the linear system defined by the control-Koopman operator.  Thus, the traditional interpretations of output energy and input energy \cite{dullerud2013course} can be articulated in terms of Koopman gramians.  In particular, given an initial condition $\psi_x(x(0)) \in\mathbb{R}^{n_L}$,  the energy of the output $y_t \in \mathbb{R}^p$ is given by 
\begin{equation}
\begin{aligned}
 ||y_t||^2_2 =  y_t^{T}y_t = \psi_x(x_0)^{T}X_{o}^\psi\psi_x(x_0),
\end{aligned}
\end{equation}
Moreover, given a target state $x_0$, the minimum input energy is quantified using the inverse of the Koopman controllability gramian.  As in classical linear systems theory, the inverse of the Koopman controllability gramian exists if the pair $(K_x,K_u)$ is controllable. 
\begin{equation}
	\begin{aligned}
\|\psi_w^{opt}\|_{2}^{2} = \psi_{x}(x_0)^{T}(X_c^\psi)^{-1}\psi_x(x_0), 
	\end{aligned}
\end{equation}
where $||\psi_w^{opt}||$ denotes the minimum disturbance energy to drive the system from $x_{-\infty}=0$ to $x_0.$

The Koopman gramians are $n_L$ by $n_L$ matrices.  In particular, not every column in the Koopman gramian coincides with a physically meaningful state.   However, leveraging  Assumption \ref{assump:stateinputinclusive}, we have that the first $n$ columns (or rows) of $X_o^\psi$ and $(X_c^\psi)^{-1}$ are directly tied to the original state $x$ contained within the observable function $\psi_x(x).$  Thus, instead of considering unit perturbations directly to the observable function $\psi(x)$ of the form 
\begin{equation}
(\psi_x(x[0]) +e_i )^T X_o (\psi_x(x[0] ) + e_i)
\end{equation}
we consider unit perturbations along the first $n$ observable functions, coinciding with {\it actual physical states} of the system.  
\begin{equation}
||\psi_x(x(t))^T \psi_x(x(t)|| =  \psi_x(x[0]+e_i )^T X_i^\psi \psi_x(x[0] +  e_i) 
\end{equation}
where $ X_i^\psi = X_o^\psi$ or $X_i=(X_c^\psi)^{-1}.$  The Koopman observability and controllability gramian are thus used to efficiently compute the energy of the input-output response, in response to unit perturbations to the input of the  nonlinear function of $\psi_x(x)$. 

Based on our decomposition criterion, for one subset $S_{i}$ of original states $x$, we define 
	\begin{equation}\label{equ:k_o}
	\begin{aligned}
	\kappa_{S_{i}}^{o} &= \frac{\psi_x(x=I_{S_{i}})^{T} X_o^{\psi} \psi_x(x=I_{S_i})}{\psi_x(x=(I-I_{S_{i}}))^{T} X_o^{\psi} \psi_x(x=(I-I_{S_i}))},
	\end{aligned}
	\end{equation}
		\begin{equation}\label{equ:k_c}
	\begin{aligned}
	\kappa_{S_{i}}^{c} &= \frac{\psi_x(x=I_{S_{i}})^{T} (X_c^{\psi})^{-1} \psi_x(x=I_{S_i}) }{\psi_x(x=(I-I_{S_{i}}))^{T}(X_c^{\psi})^{-1} \psi_x(x=(I-I_{S_i})) },	
			\end{aligned}
			\end{equation}
where $I$ is the n-dimension vector whose elements are all $1$, $I_{S_{i}}$ is the vector that states in $S_{i}$ are $1$ while the rest is 0. We can see that $\kappa_{S_{i}}^{o}$ expresses this subsystem's ability to infer its internal states compared to other subsystems given local output measurement. In the meanwhile, $\kappa_{S_{i}}^{o}$ reveals the influence of disturbances from other subsystems on subsystem $S_{i}$ compared to its own internal disturbances. We want both to be maximized simultaneously.  We define our objective function as a linearization of these two terms. 
\begin{equation}\label{equ: criterion}
	\begin{aligned}
		\kappa_{S_{i}} = \kappa_{S_{i}}^{o} + \lambda \kappa_{S_{i}}^{c}, 
	\end{aligned}
\end{equation}
where $\lambda$ is a positive constant to reveal the relative important of $\kappa_{S_{i}}^{o}$ and $\kappa_{S_{i}}^{c}$ in our concerns. Because usually $\kappa_{S_{i}}^{o}$ and $\kappa_{S_{i}}^{c}$ are with different magnitudes, we  normalize the term by its average value before we do the decomposition.  
Also, in many applications,  it is better to optimize the worst case $\min_{S_{1},\cdots,S_{k}} (\kappa_{S_{i}})$, so we get  the following optimization problem to decompose the system into $k$ subsystems:
\begin{equation}\label{equ: tt}
\begin{aligned}
\max ~~~~~~&  \min_{S_{1},\cdots,S_{k}} ~\kappa_{S_{i}} \\
\text{s.t.}~~~~~ \cup &S_{i} = x, ~\cap S_{i} = \emptyset, ~~\forall i
\end{aligned}
\end{equation}

However, there are two issues with this optimization. First, it is NP-hard that all possible set of subsystems is combinatorial number. Second, it is not a convex problem. We thus consider the following relaxations. First, suppose that $S_{i}$ and $\kappa_{S_{i}}$ both are continuous set; then we can verify \eqref{equ: tt} can be approximated by 
\begin{equation}\label{equ: t2}
\begin{aligned}
\min ~~~~~~&  \max_{S_{1},\cdots,S_{k}} ~\kappa_{S_{i}} \\
\text{s.t.}~~~~~ \cup &S_{i} = x, ~\cap S_{i} = \emptyset,~~\forall i.
\end{aligned}
\end{equation}
Note that in ideal situation, \eqref{equ: tt} and \eqref{equ: t2} will result in an equal partition which means $\kappa_{S_{i}} = \kappa_{S_{j}}, \forall i,j.$ So we can relax the original problem as follows
\begin{equation}\label{equ: t3}
\begin{aligned}
\min ~~~~~~&  \sum_{i,j}|\kappa_{S_{i}} - \kappa_{S_{j}}| \\
\text{s.t.}~~~~~ \cup &S_{i} = x, ~\cap S_{i} = \emptyset, ~~\forall i
\end{aligned}
\end{equation}
Although the $\eqref{equ: t3}$ is still NP-hard, it gives us more intuition about the nature of the problem. As we can see, the $\eqref{equ: t3}$ is trying to minimize the total difference among subsystems. That is, we want the decomposition to result in a more balanced partition, where balance is defined in terms of subsystem controllability $(\kappa_c)$ and subsystem observability ($\kappa_o$).  

This problem is equivalent to the multi-way number partitioning problem.  \cite{korf2009multi} gives an efficient and error-bounded heuristic algorithm to approximate the optimal solution. While \cite{korf2009multi} only considers minimizing the difference among sum of numbers in each subset, our problem is more complicated because when we introduce a new state into the cluster, the objective function can not be computed by directly adding the value of new state and the original rank value of the cluster, since this would require that the Koopman observable is linear in the state $x_t$.  We thus propose a revised version of the original algorithm that accounts for the nonlinearity of the Koopman observable function (see  Figure 3).
\begin{algorithm}\label{fig: 1}
	\SetKwInOut{Input}{Input}
	\SetKwInOut{Output}{Output}
	\textbf{~~~~~~~~~~~~~~K-way Partitioning function} $(a,k)$	\\
	\textbf{---------------------------------------------------------------------}\\
	\Input{Array $a$ and cluster number $k$, each $a[j]$ is the value of  $\kappa_{S_{i}}$ when $S_{i}$ only contains state $j$.}
	\Output{$k$ clusters and $\kappa_{S_{i}}^{o}$ and $\kappa_{S_{i}}^{c}$ for each cluster $S_{i}$.} 
	Initialize an empty heap $heap()$\\
	\For{$a_{j} \in a$}{
		-- Construct tuple set: $b_{i} = <a_{i},\underbrace{0,0,\cdots}_{k-1}>.$\\
		-- Define $b_{i}^{\max}$ as the first element in tuple $b_{i}$, also for each element in $b_{i}$, we will have a label to reveal if this element is a member of $x_t$ or an auxiliary Koopman observable function (not a state variable).\\
		-- $heap.push(b_{i})$.
		}
	\While{More than one tuple remains}{
		-- $b_{i}^{\max} = heap.pop()$ ~~~\#\# largest element\\
		-- $b_{j}^{\max} = heap.pop()$ ~~~\#\# second largest element\\
		\For{each permutation $p$ for $[1, \cdots,k]$}{
			-- Define $b_{new} = b_{i}^{max} \cup b_{j}^{max}(p)$, $\cup$ is element-wise operator, $p$ means that $b_{j}^{max}$ rerranged by the permutation $p$.\\
			-- Compute $\kappa_{S_{i}}$ for each dimension of $b_{new}$ and record $b_{new}^{max}$ and corresponding $p$  \\
			-- Record $b^{out}$ as the minimum of $b_{new}^{max}$ and the corresponding $p^{out}$.
		}
		$b^{in} = b_{i}^{\max} \cup b_{j}^{\max}(p^{out})$\\
		-- Compute $\kappa_{S_i}$ for each dimension of $b^{in}$ and put  value into the corresponding dimension (insert in the first slot of the array). \\
		-- Normalize the first place's value of each dimension by their minimum.
	}
	cluster = $heap.pop()$, compute $\kappa_{S_{i}}^{o}$ and $\kappa_{S_{i}}^{c}$ for each  cluster.\\
	\textbf{return} [cluster, $\kappa_{S_{i}}^{o}$,$\kappa_{S_{i}}^{o}$]\\
	\textbf{---------------------------------------------------------------------}
    \caption{Koopman gramian multi-way partitioning algorithm}
\end{algorithm}
Note that this algorithm does not consider the underlying connections in the network. For the application who has this requirement. We can slightly change the algorithm such as each time, we choose the neighborhood nodes of  the clusters to do the update.
\section{Decomposition of an IEEE-39 Bus System}
\begin{figure}
\includegraphics[width=.85\columnwidth]{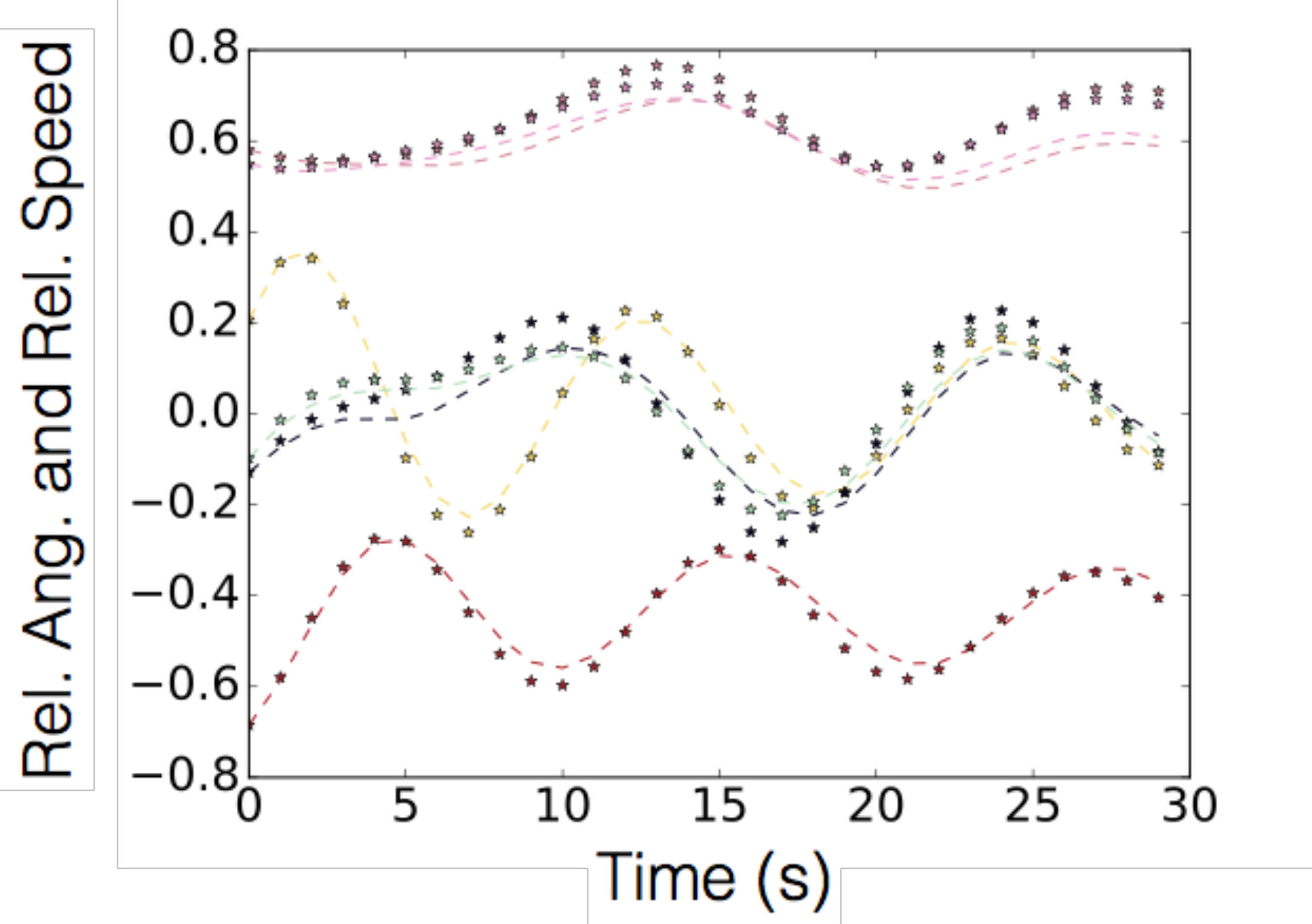}
\caption{A figure showing the multi-step time-lapse prediction (dashed lines) of a deep input-Koopman operator learned from training time-series data of the 39 bus system, using 100 random initial conditions for the relative speed and relative angle. The ground truth is plotted as dots.}\label{fig:swing_prediction}
\end{figure}
In this section we apply our Koopman controllability and observability decomposition algorithm on a swing dynamics model for the IEEE 39 bus system.   We consider the problem of modeling disturbances to the swing dynamics; specifically, we model perturbations to the relative speed and relative angle states of individual generators.  The perturbed system dynamics, assuming a disturbance $\Delta_i$ to the initial condition  of the relative angle $\delta_i(0)$ is modeled as 
\begin{equation}
\begin{aligned}
\dot{\delta_i}&=\omega_i\ + {\cal D}(\Delta_i + \delta_i(0)) = \omega_i \ + 1 u_i(t),\\
\dot{\omega_i}&=\frac{1}{M_i}\left(-D_i\,\omega_i+P_{m,i}-P_{e,i}\right) + {\cal D}(\Omega_i + \omega_i(0) )\,,\\
\end{aligned}
\end{equation}
where conservation of power flow dictates 
\begin{equation}
\begin{aligned}
P_{e,i}&=V_i\sum_{j=1}^nV_j\left(G_{ij}\cos(\delta_i-\delta_j)+B_{ij}\sin(\delta_i-\delta_j)\right)\,.
\end{aligned}
\end{equation}
where $G_{ij}$ and $B_{ij}$ are the transfer conductance and susceptance between buses $i$ and $j$, and $i = 1,...,n$, where $n =9.$  Here we have taken generator 10 as a reference frame for the relative speed $\omega_i$ and relative angle $\delta_i$ of each system. 
The nonlinearity of the power flow equations couple directly into the swing dynamics, defining an implicit nonlinear dynamical system with sinusoidal terms.   

We consider a Koopman operator approach, motivated by several observations \cite{MezicPower}.  First, our approach is data-driven, which is suitable for wide-area monitoring applications when working with uncalibrated or parametrically varying models (due to changes in operational context or uncertainty).  Second, the approach is potentially scalable, since an increase in the number of system states can be captured using larger Koopman dictionaries, generated by scalable deep learning algorithms in Tensorflow.  Third, Koopman modes naturally reveal coherency or incoherency among generators \cite{susuki2011nonlinear, MezicPower}, which can be useful for distributed control design and transient stability analysis\cite{mezic2005spectral,mauroy2016global}. 

We trained an input-Koopman operator for the 39 bus using our deep dynamic mode decomposition algorithm.  We implemented a 20-wide, 10 layer deep ResNet with exponential linear unit activation functions (ELUs) \cite{clevert2015fast} and Dropout \cite{srivastava2014dropout} in TensorFlow.  We generated 100 random trajectories of the IEEE 39 bus swing dynamics (1 second resolution) and used 50 of the trajectories as training data for learning a deep Koopman operator.  The remaining 50 trajectories were reserved for the test stage.   All trajectories were defined on a normalized scale, to ensure convergence of the Adam optimizer.  Model accuracy was evaluated in two ways: 1) one-step prediction error on test data (less than $0.01\% $) and 2) multi-step prediction error on a previously unseen initial condition drawn from the test data (less than $0.4\%$ error per time-step).  The results of the multi-step prediction task are plotted in Figure \ref{fig:swing_prediction}.

\begin{figure}
\centering
\includegraphics[width=.8\columnwidth]{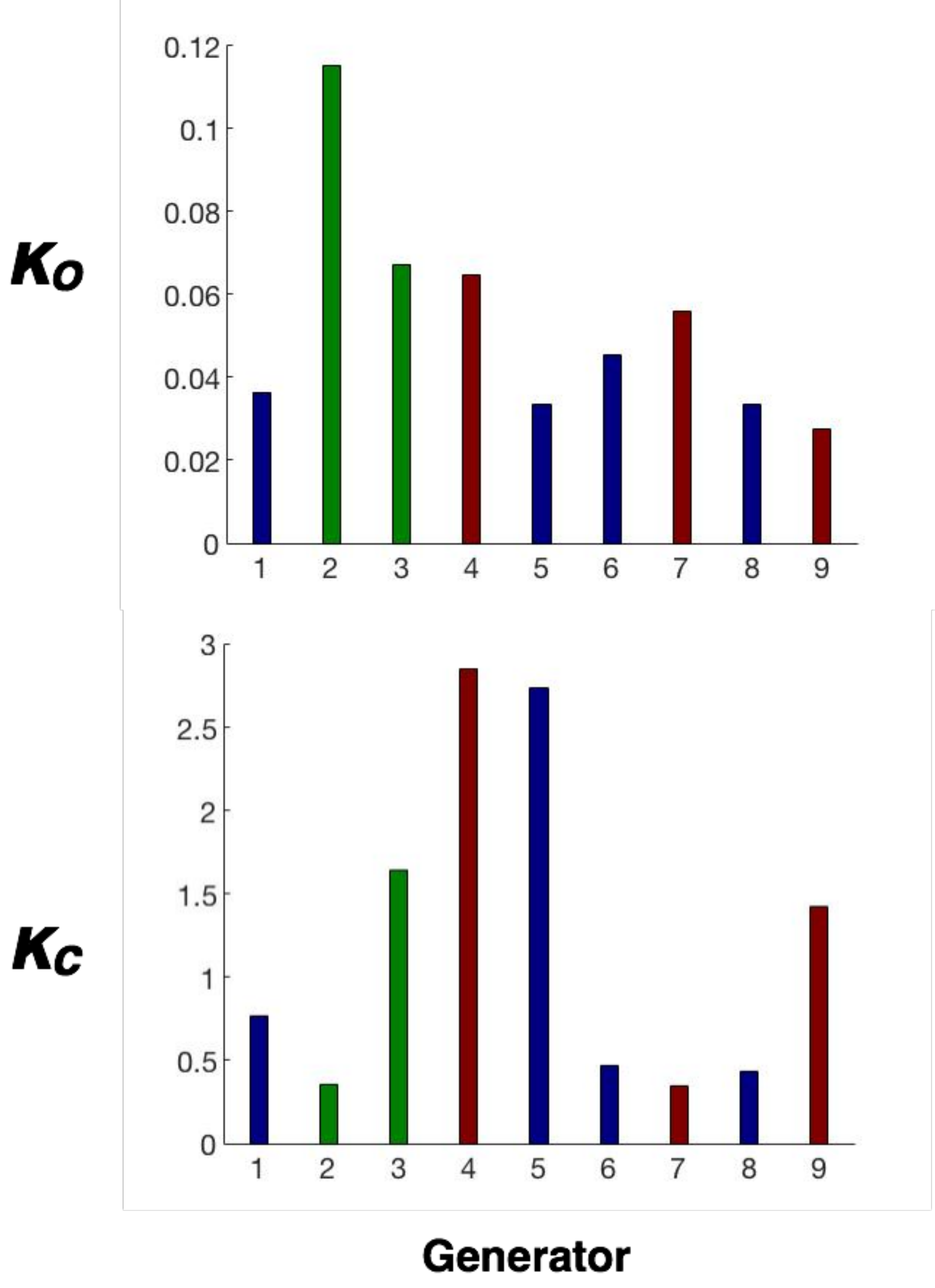}
\caption{Bar charts of $\kappa_o$ (subsystem relative observability) and $\kappa_c$ (subsystem relative controllability) for each generator.}\label{fig:kappa_bar}
\end{figure}
We next applied our decomposition algorithm to compute a partition on the generators (see Figure \ref{fig:clusters}) that accounted for balancing of $\kappa_o$ (subsystem observability) and $\kappa_c$ (subsystem controllability).  This decomposition results in subsystems that maximize internal observability while simultaneously maximizing resiliency to external disturbances.    The raw values of $\kappa_o$ and $\kappa_c$ for each generator are plotted in Figure \ref{fig:kappa_bar}.  We noted that generator 2 exhibited high $\kappa_o$ scores, but had a low $\kappa_c$ score, indicating susceptibility to external disturbances from other generators in the swing equations. 

Our algorithm identified a clustering of generators resulting in a maximum variation in $\kappa_c$ of \[\max_{i,j}| \kappa_{c,i} - \kappa_{c,j}| = 0.33434 \] and a maximum variation in $\kappa_{o}$ of \[\max_{i,j} |\kappa_{o,i} - \kappa_{o,j}| = 0.09744.\]   Since our algorithm normalized the contributions of both $\kappa_c$ and $\kappa_o$, we were able to balance despite the separation in scales.  Interestingly, our algorithm identified zonal clusterings for generators that were for the most part closer to each other than generators outside their cluster (see Figure \ref{fig:clusters}).  However, this was not always the case; generator 6 exhibited a weak $\kappa_c$ score and relatively weak $\kappa_o$ score and thus was included in a separate clustering to ensure balanced subsystem controllability and observability across clusters.  These observations illustrate the input-output focus of our Koopman decomposition approach. Whereas area assignments are often allocated based on spatial criteria, our algorithm seeks to minimize subsystem sensitivity to disturbances from other subsystems and internal observability, to ensure accurate state estimation \cite{korda2016linear}. 

\begin{figure}
\centering
\includegraphics[width=.9\columnwidth]{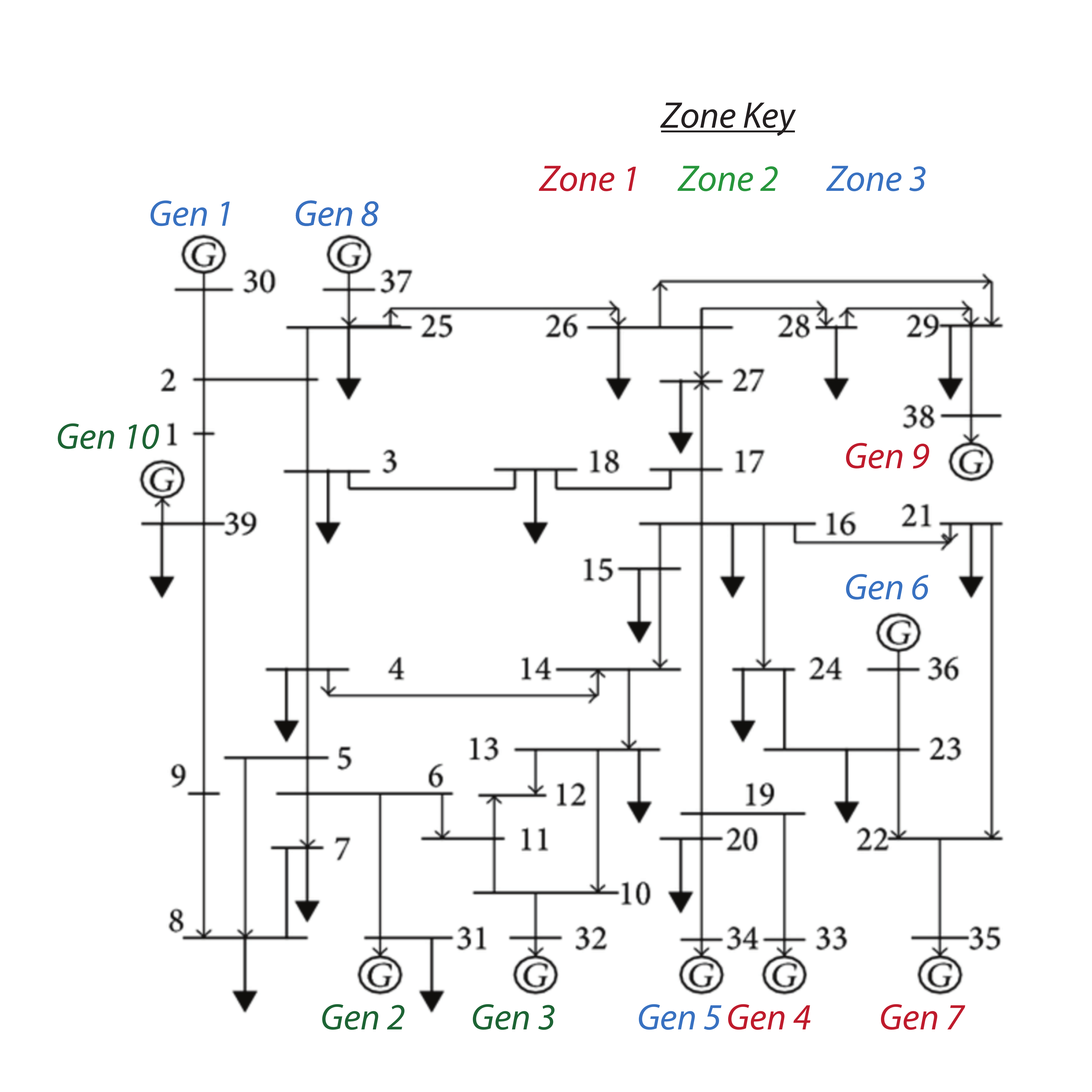}
\caption{A figure showing the area/zone assignment of generators by our modified multi-way partitioning algorithm, based on Koopman controllability and observability.}\label{fig:clusters}
\end{figure}

\section{Conclusion}
In this paper we introduced a method  for learning input-Koopman operators from data by integrating deep learning and dynamic mode decomposition.  In particular, we explored the conditions under which a input-state separable form exists for the input-Koopman equations, enabling construction of Koopman gramians.   We showed that deep dynamic mode decomposition is able to recover a high fidelity approximation involving input-state separable Koopman operators, even when the underlying system is not input-state separable.  These results underscore the power of learning representations for approximate data-driven control. We introduced a nonlinear decomposition algorithm, based on Koopman gramians, that maximizes internal subsystem observability and disturbance rejection from unwanted noise in other subsystems.  Future work will focus on scalability and analysis of system with hybrid dynamics. 

 	\bibliographystyle{unsrt}
 	\bibliography{./Section/references}
\end{document}